\documentclass[submission,copyright,creativecommons]{eptcs}
\providecommand{\event}{QPL 2015} 

\usepackage[centertags,sumlimits,intlimits,namelimits,reqno]{amsmath}
\usepackage{latexsym,amsfonts,amssymb,amsthm,exscale,enumerate,breakcites}
\usepackage[all]{xy}
\let\originallhook\lhook
\usepackage{MnSymbol}
\newcommand{\lhook}{\mathrel{\raise.018ex\hbox{$\originallhook$}}}
\usepackage[T1]{fontenc}
\usepackage{tikz}

\pgfdeclarelayer{edgelayer}
\pgfdeclarelayer{nodelayer}
\pgfsetlayers{edgelayer,nodelayer,main}
\tikzstyle{none}=[inner sep=0mm]

\newtheorem{defn}{Definition}
\newtheorem{thm}[defn]{Theorem}
\newtheorem{prop}[defn]{Proposition}
\newtheorem{lem}[defn]{Lemma}
\newtheorem{cor}[defn]{Corollary}
\theoremstyle{remark}

\newtheorem{rmk}[defn]{Remark}

\newcommand{\hooklongrightarrow}{\lhook\joinrel\longrightarrow}

\newcommand{\Ca}{\mathbf{C}}
\newcommand{\Cafree}{\Ca_{\textrm{free}}}
\newcommand{\re}{\mathbb{R}}

\newcommand{\N}{\mathbb{N}}

\newcommand{\Set}{\mathbf{Set_\times}}
\newcommand{\SetU}{\mathbf{Set}}
\newcommand{\InjU}{\mathbf{Inj}}
\newcommand{\BijU}{\mathbf{Bij}}
\newcommand{\Rel}{\mathbf{Rel_\times}}

\DeclareMathOperator{\mor}{Mor}
\DeclareMathOperator{\pcd}{PCD}

\providecommand{\urlalt}[2]{\href{#1}{#2}} 
\providecommand{\doi}[1]{doi:\urlalt{http://dx.doi.org/#1}{#1}}

\title{Additive monotones for resource theories of parallel-combinable processes
with discarding}
\author{Brendan Fong
\institute{Department of Computer Science \\ University of Oxford}
\email{brendan.fong@cs.ox.ac.uk}
\and
Hugo Nava-Kopp
\institute{Department of Computer Science \\ University of Oxford}
\email{hugo.nava.kopp@cs.ox.ac.uk}
}
\def\titlerunning{Additive monotones for resource theories of parallel-combinable processes
with discarding}
\def\authorrunning{B.\ Fong \& H.\ Nava-Kopp}
\begin{document}
\maketitle

\begin{abstract}
  A partitioned process theory, as defined by Coecke, Fritz, and Spekkens, is a
  symmetric monoidal category together with an all-object-including symmetric
  monoidal subcategory. We think of the morphisms of this category as processes,
  and the morphisms of the subcategory as those processes that are freely
  executable.  Via a construction we refer to as parallel-combinable processes
  with discarding, we obtain from this data a partially ordered monoid on the
  set of processes, with $f \succeq g$ if one can use the free processes to
  construct $g$ from $f$. The structure of this partial order can then be probed
  using additive monotones: order-preserving monoid homomorphisms with values in
  the real numbers under addition. We first characterise these additive
  monotones in terms of the corresponding partitioned process theory.

  Given enough monotones, we might hope to be able to reconstruct the order on
  the monoid. If so, we say that we have a complete family of monotones. In
  general, however, when we require our monotones to be additive monotones, such
  families do not exist or are hard to compute. We show the existence of
  complete families of additive monotones for various partitioned process
  theories based on the category of finite sets, in order to shed light on the
  way such families can be constructed.
\end{abstract}

\section{Introduction}

In \cite{CFS2}, Coecke, Fritz, and Spekkens make a well-illustrated case for
viewing symmetric monoidal categories as theories of resources: the objects of
the category are interpreted as resources, the morphisms methods for converting
one resource into another. In many examples, such as quantum entanglement, the
resources themselves are processes, and the methods of converting one process
into another involve composition with a set of `freely executable' processes.
This structure is formalised as a partitioned process theory: a pair comprising
a symmetric monoidal category together with an all-object-including symmetric
monoidal subcategory.

The question then arises: when can we build one resource from another? One
technique for answering this question is to assign real numbers to each
resource, according to their power to create other resources. We call such
functions monotones. A collection of monotones that completely characterises
this structure is known as a complete family of monotones. While complete
families of monotones always exist, ones with nice properties are in general
hard to come by, with even celebrated ones, such as entropy or Kolmogorov
complexity, being difficult to compute.

In this article we explore the construction of families of so-called additive
monotones---monotones that are also monoid homomorphisms into the real numbers
under addition. Our main result fixes a particular method of building a resource
theory from a partitioned process theory, and characterises the additive
monotones on the resulting resource theory. We then explore two applications of
this theorem, using it to construct complete families of additive monotones.

\section{Resource theories of parallel-combinable processes with discarding}

We formalise our ideas about resources and free processes using symmetric
monoidal categories, following the work of Coecke, Fritz, and Spekkens
\cite{CFS1, CFS2}. We say that a subcategory $\mathbf D$ of a category $\Ca$
is a \emph{wide} subcategory if it includes all the objects of $\Ca$.

\begin{defn}\label{def:ResourceTheory} 
  A \emph{partitioned process theory} $(\Ca, \Cafree)$ consists of a symmetric
  monoidal category $\Ca$ together with a wide symmetric monoidal subcategory
  $\Cafree$.
\end{defn}

With this definition we see that examples of partitioned process theories
abound, not only describing structures such as entanglement and athermality
arising in applied sciences \cite{BHORS, HHHH}, but also simply natural ideas
within mathematics itself. For example, any category $\Ca$ with products can be
considered a symmetric monoidal category with monoidal product given by the
categorical product, and in such a category each object is equipped with a
commutative comonoid structure. We may generate a wide symmetric monoidal
subcategory from these comonoid morphisms to serve as our category $\Cafree$. A
similar, topical \cite{Ki,Fo,Mo}, example arises from the analogous construction
on so-called multigraph categories---categories in which every object is
equipped with a special commutative Frobenius monoid.

We caution that a partitioned process theory was simply termed a resource theory
in \cite{CFS1}; our terminology comes from \cite{CFS2}, and following
\cite{CFS2} we instead use \emph{resource theory} to simply refer to a symmetric
monoidal category in which we think of the objects as resources. In line with
this viewpoint, we shall refer to the morphisms of $\Ca$ as \emph{processes} and
the morphisms of $\Cafree$ as \emph{free processes}. We then can construct a
resource theory from a partitioned process theory by considering the processes
as resources.

Indeed, given a partitioned process theory $(\Ca,\Cafree)$, we can construct a
symmetric monoidal category in which the objects are the processes $\Ca$, and
the morphisms are methods of constructing one process from another using free
processes. We term this new category the \emph{resource theory of
parallel-combinable processes with discarding} $\mathrm{PCD}(\Ca,\Cafree)$ of
the partitioned process theory $(\Ca,\Cafree)$.

We shall assume all categories here are small, and write $\lvert \Ca \rvert$ and
$\mor(\Ca)$ for the sets of objects and morphisms of a category $\Ca$
respectively.

\begin{prop}\label{def:parallelResourceDefinition}
  Let $(\Ca,\Cafree)$ be a partitioned process theory. Then we may define a
  symmetric monoidal category $\pcd(\Ca,\Cafree)$ with objects $f \in \mor(\Ca)$, and
  morphisms $f \to g$ triples $(Z,\xi_1,\xi_2)$, where $Z \in
  \lvert\Ca\rvert$, $\xi_1,\xi_2 \in \mor(\Cafree)$, such that there exists $j
  \in \mor(\Ca)$ with 
  \begin{equation} \label{eq:parallelResourceEquation}
    \xi_2 \circ (f \otimes 1_Z) \circ \xi_1 = g \otimes j.
  \end{equation}
  In string diagrams equation (\ref{eq:parallelResourceEquation}) becomes:
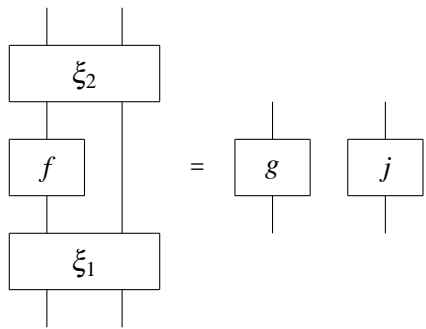
  \[
    \begin{tikzpicture}
      \begin{pgfonlayer}{nodelayer}
	\node [style=none] (0) at (-0.5, 1.25) {};
	\node [style=none] (1) at (-0.5, 0.75) {};
	\node [style=none] (2) at (0.5, -0.5) {};
	\node [style=none] (3) at (0.5, 1.25) {};
	\node [style=none] (4) at (-1, 2) {};
	\node [style=none] (5) at (-1, 1.25) {};
	\node [style=none] (6) at (1, 1.25) {};
	\node [style=none] (7) at (1, 2) {};
	\node [style=none] (8) at (-0.5, 2) {};
	\node [style=none] (9) at (0.5, 2) {};
	\node [style=none] (10) at (-1, 0.75) {};
	\node [style=none] (11) at (0, 0.75) {};
	\node [style=none] (12) at (-0.5, 2.5) {};
	\node [style=none] (13) at (0.5, 2.5) {};
	\node [style=none] (14) at (0, -0) {};
	\node [style=none] (15) at (-1, -0) {};
	\node [style=none] (16) at (-1, -0.5) {};
	\node [style=none] (17) at (1, -0.5) {};
	\node [style=none] (18) at (1, -1.25) {};
	\node [style=none] (19) at (-1, -1.25) {};
	\node [style=none] (20) at (-0.5, -0.5) {};
	\node [style=none] (21) at (-0.5, -0) {};
	\node [style=none] (22) at (-0.5, -1.25) {};
	\node [style=none] (23) at (-0.5, -1.75) {};
	\node [style=none] (24) at (0.5, -1.25) {};
	\node [style=none] (25) at (0.5, -1.75) {};
	\node [style=none] (26) at (0, 1.5) {};
	\node [style=none] (27) at (1.5, 0.25) {};
	\node [style=none] (28) at (2, -0) {};
	\node [style=none] (29) at (3, -0) {};
	\node [style=none] (30) at (2, 0.75) {};
	\node [style=none] (31) at (3, 0.75) {};
	\node [style=none] (32) at (3.5, 0.75) {};
	\node [style=none] (33) at (3.5, -0) {};
	\node [style=none] (34) at (4.5, -0) {};
	\node [style=none] (35) at (4.5, 0.75) {};
	\node [style=none] (36) at (2.5, 0.75) {};
	\node [style=none] (37) at (4, 0.75) {};
	\node [style=none] (38) at (4, -0) {};
	\node [style=none] (39) at (2.5, -0) {};
	\node [style=none] (40) at (2.5, 1.25) {};
	\node [style=none] (41) at (4, 1.25) {};
	\node [style=none] (42) at (2.5, -0.5) {};
	\node [style=none] (43) at (4, -0.5) {};
	\node [style=none] (44) at (0, 1.6) {$\xi_2$};
	\node [style=none] (45) at (0, -0.9) {$\xi_1$};
	\node [style=none] (46) at (-0.5, 0.35) {$f$};
	\node [style=none] (47) at (2.5, 0.35) {$g$};
	\node [style=none] (48) at (4, 0.35) {$j$};	
	\node [style=none] (49) at (1.5, 0.35) {$=$};	

      \end{pgfonlayer}
      \begin{pgfonlayer}{edgelayer}
	\draw (0.center) to (1.center);
	\draw (3.center) to (2.center);
	\draw (4.center) to (5.center);
	\draw (5.center) to (6.center);
	\draw (6.center) to (7.center);
	\draw (7.center) to (4.center);
	\draw (10.center) to (11.center);
	\draw (14.center) to (11.center);
	\draw (14.center) to (15.center);
	\draw (15.center) to (10.center);
	\draw (8.center) to (12.center);
	\draw (9.center) to (13.center);
	\draw (16.center) to (17.center);
	\draw (18.center) to (17.center);
	\draw (18.center) to (19.center);
	\draw (19.center) to (16.center);
	\draw (23.center) to (22.center);
	\draw (25.center) to (24.center);
	\draw (20.center) to (21.center);
	\draw (30.center) to (31.center);
	\draw (29.center) to (31.center);
	\draw (29.center) to (28.center);
	\draw (28.center) to (30.center);
	\draw (32.center) to (35.center);
	\draw (33.center) to (34.center);
	\draw (34.center) to (35.center);
	\draw (33.center) to (32.center);
	\draw (36.center) to (40.center);
	\draw (37.center) to (41.center);
	\draw (42.center) to (39.center);
	\draw (43.center) to (38.center);
      \end{pgfonlayer}
    \end{tikzpicture}
  \]
\end{prop}
\begin{proof}
A proof of this proposition may be found in \cite[Theorem
11]{CFS1}.\footnote{While it makes little difference for our purposes here, to
  do away with this assumption that the categories must be small, we may define
  a notion of equivalence on the triples $(Z,\xi_1,\xi_2)$ so that we consider
  two triples to be equivalent if they have, roughly speaking, the same
  `operational behaivour'. For similar constructions this approach is favoured
  in \cite[Section 3]{CFS2}.}
\end{proof}

\begin{rmk}
  Note that there are many ways to construct interesting resource theories from a
  partitioned process theory, depending on the methods we allow for turning one
  process into another using free processes. Other examples include the resource
  theory of parallel-combinable processes $\operatorname{PC}(\Ca,\Cafree)$
  \cite[Subsection 3.3]{CFS2} and the resource theory of universally-combinable
  processes $\operatorname{UC}(\Ca,\Cafree)$ \cite[Subsection 3.4]{CFS2}. While we
  shall not define these constructions here, for the reader familiar with them, we
  note that we have the inclusion of symmetric monoidal categories
  \[
    \operatorname{PC}(\Ca,\Cafree) \hooklongrightarrow \pcd(\Ca,\Cafree).
  \]
  Moreover, when the partitioned process theory obeys certain conditions, it can be
  shown that these definitions coincide, although in general they do not. We also
  note that it is possible to interpret these different constructions as different
  methods of constructing an operad of `wiring diagrams' \cite{Sp} from the set of
  free processes. 
\end{rmk}

In this paper we are interested in understanding the ordered monoid
corresponding to this resource theory. Recall that an \emph{ordered monoid}
$(X,\succeq,\cdot)$ is a set $X$ together with a partial order $\succeq$ and a monoid
multiplication $\cdot$ such that for all $x,y,z,w \in X$,
\begin{equation} \label{eq:orderedmonoid}
  \mbox{if}\quad x \succeq y \mbox{ and }
  z \succeq w, \quad\mbox{ then} \quad x\cdot z \succeq y \cdot w.
\end{equation}
We may partially decategorify a resource theory to obtain an ordered monoid in
the following way. 

\begin{thm}\label{thm:PreorderedMonoid}
Let $\mathbf R$ be a symmetric monoidal category, and call objects $f, g$ in
$\mathbf R$ equivalent if there exists morphisms $f \to g$ and $g \to f$. This
defines an equivalence relation. Write $[f]$ for the equivalence class of the
object $f$; we shall frequently abuse notation to simply write $f$ for the
equivalence class of $[f]$ and $\lvert\mathbf R\rvert$ for the set of
equivalence classes of objects in $\mathbf R$. 

Then there exists an ordered monoid $(\lvert\mathbf R\rvert,\succeq,\otimes)$ on
the set of these equivalence classes, with $[f] \succeq [g]$ if there exists a
morphism $f \to g$ in $\mathbf R$, and using the monoidal product in $\mathbf R$
to define $[f] \otimes [g] = [f \otimes g]$. Moreover, this monoid is commutative.
\end{thm}
\begin{proof}
  The relation on the objects of $\mathbf R$ specified by the existence of a
  morphism $f \to g$ is reflexive due to identity morphisms, and the transitive
  due to composition. Thus we obtain a partial order on the set of equivalence
  classes of $\lvert\mathbf R\rvert$. 
  
  The unit for the monoid $\lvert\mathbf R\rvert$ is given by the monoidal unit
  in $\mathbf R$; the unitors show that it is indeed a unit for the monoid
  multiplication. Similarly, the associativity of the monoid multiplication
  follows from the existence of the associators for the monoidal product, and
  the commutativity from the braiding. Moreover, the compatibility condition
  (\ref{eq:orderedmonoid}) follows from the functoriality of the monoidal
  product. 
\end{proof}

Under this equivalence relation, we consider two resources the same if we may
convert one into the other, and vice versa. We then think of this ordered monoid
as a \emph{theory of resource convertibility}, with the monoid structure
describing how we can combine two resources to make another, and the partial
order describing when we can turn one resource (more precisely, equivalence class of
resources) into another.

Note that in the theory of resource convertibility for a resource theory of
parallel-combinable processes with discarding, the free morphisms themselves
form an equivalence class, and that this equivalence class acts as the identity
element for the monoid multiplication.  Note also that this equivalence class
contains the identity morphisms $1_X$ for all $X$, so there is no confusion to
be had by writing this equivalence class as $1$. 

By `discarding' in `resource theory of parallel-combinable processes with
discarding' we mean that no cost is incurred by replacing some resource $g
\otimes j$ with just some subpart $g$ of it. In terms of theories of resource
convertibility, this means that the corresponding monoid is non-negative. Recall
that we call an ordered monoid $(X,\ge,\cdot)$ \emph{non-negative} if the
identity element $1$ of the monoid is the bottom element for the partial order;
that is, if we have 
\[
  \mbox{ for all } x \in X,\enspace x \ge 1.
\]
This is also equivalent to the property that for all $x,y \in X$ we have $x
\cdot y \ge y$.  The ordered monoid corresponding to a resource theory of
parallel-combinable processes with discarding is always non-negative.

\begin{lem}
  Let $(\Ca,\Cafree)$ be a partitioned process theory. Then the ordered monoid
  $(\lvert \pcd(\Ca,\Cafree)\rvert, \succeq, \otimes)$ is non-negative.
\end{lem}
\begin{proof}
  Given $\xi: A \to B \in \mor(\Cafree)$ and $f: X \to Y \in \mor(\Ca)$, we see
  that $\sigma^{-1}_{Y,B} \circ (f \otimes
  1_B) \circ (1_X \otimes \xi) \circ \sigma_{A,X} = \xi \otimes f$, where
  $\sigma_{A,X}: A \otimes X \to X \otimes A$ is the braiding. Thus $[f] \succeq
  [\xi] = [1]$.  
\end{proof}

\begin{rmk}
Although \cite{CFS1} demonstrates that the resource theory of
parallel-combinable processes with discarding is a highly applicable structure,
there are some instances in which the discarding means it does not yield an
interesting ordered monoid. 

For example, consider the partitioned process theory $(\Rel, \Set)$, where
$\Rel$ is the symmetric monoidal category with objects finite sets, morphisms
relations, and monoidal product cartesian product, and $\Set$ is the wide
symmetric monoidal subcategory with morphisms restricted to the functions. We
might give the partitioned process theory $(\Rel,\Set)$ the following
interpretation. Each relation $f \subseteq X \times Y$ may be viewed as a noisy
possibilistic communication channel, with $x \in X$ possibly mapping to any of
the $y \in Y$ such that $(x,y) \in f$. The free morphisms, in this case
functions, represent encoding and decoding functions, while the cartesian
product models the fact that any channel can be used arbitrarily many times. The
resource theory $\pcd(\Rel,\Set)$ might thus be an attempt at a model for
simulatability of one possibilistic channel by another.

The ordered monoid corresponding to this resource theory, however, is trivial.
Write the empty set $\varnothing$ and let $1$ be a singleton set. Given
relations $f: X \to Y$ and $g: A \to B$, we may choose $Z = \varnothing$,
$\xi_1: \varnothing \to X$, $\xi_2: Y \to B$, and $j: \varnothing \to 1$. Then
both $\xi_2 \circ (f \times 1_Z) \circ \xi_1$ and $g \times j$ are the unique
relation $\varnothing \to B$. This implies that for all $f,g \in \mor(\Rel)$, we
have a morphism $f \to g \in \pcd(\Rel,\Set)$, thus yielding a trivial monoid.
\end{rmk}

\section{Complete families of monotones}

In information theory, one is often interested in the trying to define an
entropy of a source. These entropies give a real number quantifying, loosely
speaking, the randomness of the source. Such functions provide insight into
whether one source might be simulated by another \cite{CT}. In terms of
partitioned process theories $(\Ca,\Cafree)$, this suggests we might look at
order-preserving functions from the ordered monoid derived from
$\pcd(\Ca,\Cafree)$ to the real numbers. We call such functions monotones.

\begin{defn}
  Let $(X,\succeq)$ be a partially ordered set. A \emph{monotone} is an
  order-preserving function $M: (X, \succeq) \to (\re, \ge)$. It is further
  called \emph{complete} if it is also order-reflecting: that is, if for all $x,
  y \in X$ we have 
  \[
    x \succeq y \quad \mbox{ if and only if } \quad 
    \enspace M(x) \ge M(y).
  \]
\end{defn}

A complete monotone exists for an partially ordered set if and only if it embeds
into the reals. This is rarely the case. It is always possible, however, to find
a \emph{complete family} of monotones \cite[Proposition 5.2]{CFS2}.

\begin{defn}
  Given a partially ordered set $(X,\succeq)$, we call a collection $\{M_i\}_{i
  \in I}$ of monotones on $(X, \succeq)$  a \emph{complete family of monotones}
  if for all $x,y \in X$ we have
  \[
    x \succeq y \quad \mbox{ if and only if } \quad 
    \enspace M_i(x) \ge M_i(y) \mbox{ for all } i \in I.
  \]
\end{defn}

Some families, however, are better than others. To provide additional insight
into the structure of the ordered monoid, we frequently require some extra
properties, such as preservation of some sort of monoid structure.

\begin{defn}
  We say that a monotone on an ordered set $(X,\succeq,\cdot)$ is
  \begin{itemize}
    \item \emph{additive} if it is also a monoid homomorphism into
      $(\re,\ge,+)$.
    \item \emph{non-negative} if its image in $\re$ forms a non-negative ordered
      monoid.
  \end{itemize}
\end{defn}

Additive monotones are also known as \emph{functionals} \cite[Section 3]{Fr}.
One advantage of working with additive monotones is that we can use them to
recover the structure of the commutative monoid. Indeed, a complete family of
monotones $\{M_i\}_{i \in I}$ on $(X,\succeq,\cdot)$ induces an injective
order-preserving monoid homomorphism into $(\re^I,\ge,+)$, where the partial
order $\ge$ and monoid multiplication $+$ are given by the pointwise
construction, and so locates $(X,\succeq,\cdot)$ as a sub-ordered monoid of
$\re^I$. 

We shall work towards providing some examples of complete families of additive
monotones for resource theories of parallel-combinable processes with
discarding. Our next step is to characterise the non-negative additive monotones
for such resource theories. 

\begin{thm} \label{thm:main}
  Let $(\Ca,\Cafree)$ be a partitioned process theory and let $(X,\ge,\cdot)$ be
  a non-negative ordered monoid. A function $\mu:\mor(\Ca)
  \to X$ induces an order-preserving monoid homomorphism 
  \begin{align*}
    M:(\lvert \pcd(\Ca,\Cafree)\rvert, \succeq, \otimes) &\longrightarrow
    (X,\ge,\cdot) \\
    [f] &\longmapsto \mu(f) 
  \end{align*}
  if and only if for all
  $Z \in \lvert \Ca \rvert$, $f,g \in \mor(\Ca)$, and $\xi \in \mor(\Cafree)$ we
  have
  \begin{enumerate}[(i)]
    \item $\mu(f \otimes g) = \mu(f)\cdot\mu(g)$;
    \item $\mu(1_Z)=1$; and
    \item $\mu(f) \ge \mu(\xi \circ f)$ and $\mu(f) \ge \mu(f \circ
      \xi)$ whenever such composites are well-defined.
  \end{enumerate}
  Moreover, this gives a one-to-one correspondence: every order-preserving
  monoid homomorphism on $(\lvert\pcd(\Ca,\Cafree)\rvert,\succeq,\otimes)$
  arises from a unique such function $\mu$.
\end{thm}

Conditions \textit{(i)} and \textit{(ii)} are unsurprising; they simply ask that
the function respect the monoid structure. Condition \textit{(iii)}, however, is a
bit more illuminating: it tells us that as long as composition with free
morphisms reduces the value of the function $\mu$, the function induces a monotone.

\begin{proof}
  Suppose first that $\mu$ induces an ordered monoid homomorphism $M$. Note in
  particular this means that $M$ is well-defined on the equivalence classes of
  objects in $\pcd(\Ca,\Cafree)$, and also that $M$ preserves the order and
  monoid multiplication. Given that $M$ preserves the monoid multiplication, we
  have
  \[
    \mu(f\otimes g) = M([f\otimes g]) = M([f])\cdot M([g]) = \mu(f)\cdot\mu(g),
  \]
  and as $M$ preserves identities, we have $\mu(1_Z) = M([1_Z]) = 0$. Recall now
  that $f \succeq g$ if and only if there exists free processes $\xi_1, \xi_2$,
  an object $Z$, and a morphism $j$ in $\Ca$ such that $\xi_1 \circ (f \otimes
  1_Z) \circ \xi_2 = g \otimes j$. Then, whenever they are defined, we thus have
  $f \succeq f \circ \xi$ and $f \succeq \xi \circ f$, and so since $M$ is
  monotone we have
  \[
    \mu(f) = M([f]) \ge M([f \circ \xi]) = \mu(f \circ \xi)
  \]
  and similarly $\mu(f) \ge \mu(\xi \circ f)$. Thus $\mu$ obeys
  \textit{(i)--(iii)}.

  Conversely, suppose that $\mu$ has the properties \textit{(i)--(iii)}.  
  Now suppose that we have processes $f$ and $g$ such that $\xi_1 \circ (f
  \otimes 1_Z) \circ \xi_2 = g \otimes j$. We thus have
  \[
    \mu(f) \stackrel{\textit{(i,ii)}}{=} \mu(f \otimes 1_Z)
    \stackrel{\textit{(iii)}}{\ge} \mu(\xi_1 \circ (f \otimes 1_Z) \circ \xi_2)
    = \mu(g \otimes j) \stackrel{\textit{(i)}}{\ge} \mu(g).
  \]
  This implies that if there exists a morphism $f \to g$ in $\pcd(\Ca,\Cafree)$
  then $\mu(f) \ge \mu(g)$, and hence that $M([f]) = \mu(f)$ is well-defined
  and monotone. Properties \textit{(i)} and \textit{(ii)} then immediately
  imply $M$ is a monoid homomorphism.

  Finally, given a homomorphism of ordered monoids $M: (\lvert
  \pcd(\Ca,\Cafree)\rvert,\succeq,\otimes) \to (X,\ge,\cdot)$, we may define
  $\mu(f) = M([f])$ to obtain the unique function $\mu: \mor(\Ca) \to X$ that
  induces it.
\end{proof}

This theorem allows us to construct non-negative additive monotones by working
from the partitioned process theory $(\Ca,\Cafree)$. We take advantage of this
fact in the next section to produce some families of complete monotones.

\section{Example: encoding functions as resources} \label{sec:ex}

In this section we construct two examples of complete families of non-negative
additive monotones for resource theories of parallel-composable processes with
discarding. Both partitioned process theories at hand have as processes the
functions between finite sets. They may hence be understood as modelling
encoding schemes, where a resource is a method for assigning a code symbol to
each element of some finite input set. Our ordered monoids thus answer the
question of when we may use the free morphisms and the construction
(\ref{eq:parallelResourceEquation}) to turn one encoding scheme into another.

We emphasise here the concreteness of these results: given two functions between
finite sets, one can quickly use the complete families of additive monotones we
construct to evaluate whether one resource is convertible into the other.

\subsection{The partitioned process theory of functions and bijections}\label{sec:SetuPermu}

Let $\SetU$ the symmetric monoidal category with objects finite sets, morphisms
functions, and monoidal product disjoint union, and let $\BijU$ be the wide
symmetric monoidal subcategory with morphisms restricted to the bijective
functions. Write $\# X$ for the cardinality of a set $X$.

\begin{prop}\label{def:measureBijuSetU} 
  For $i \in \N$, define functions:
  \begin{align*}
    \varphi_{i}: \mor(\SetU) &\longrightarrow \N;  \\
    (f: X \to Y) &\longmapsto \#\{y \in Y \mid \# f^{-1}(y) =i\}.
  \end{align*}
  The family of monotones $\{F_i\}_{i \in \N\setminus\{1\}}$ induced by the
  family of functions $\{\varphi_i\}_{i \in \N\setminus\{1\}}$
  is a complete family of additive monotones for the resource theory
  $\pcd(\SetU,\BijU)$.
\end{prop}
  Observe that each $\varphi_i$ takes a function $f: X \to Y$ and returns the
  number of elements of $Y$ that have $i$ elements of $X$ map to it. Also note
  that for every list of $\ell: \N \to \N$ of natural numbers with only finitely
  many entries nonzero, there exists a function $j_\ell:C \to D \in \mor(\SetU)$
  such that $\varphi_i(j_\ell) = \ell(i)$---indeed, simply choose $C_i$ of
  cardinality $i \cdot \ell(i)$, $D_i$ of cardinality
  $\ell(i)$, let $j_{\ell,i}$ map $i$ elements of $C_i$ to each element of $D$,
  and then define $j_\ell = \bigsqcup_{i = 0}^\infty j_{\ell,i}$. 
\begin{proof}[Proof of Proposition \ref{def:measureBijuSetU}]
  That we have defined a family of additive monotones is an immediate
  consequence of Theorem \ref{thm:main}. Remembering that $i \ne 1$ and that the
  free morphisms in this case are the bijections, it is clear that each
  $\varphi_i$ has the properties \textit{(i)--(iii)} required.
 
  We turn our attention to completeness. Fix functions $f: X \to Y$ and $g: A
  \to B$. For completeness we need to show that if for all $i \in \N \setminus\{1\}$ we
  have $F_i(f) \geq F_i(g)$, then there exists a finite set $Z$, bijections
  $\xi_1,\xi_2$, and a function $j$ such that
  \[
    \xi_2 \circ (f \sqcup 1_Z) \circ \xi_1 = g \sqcup j.
  \]
  Let us construct such data as follows. Consider the list of
  numbers $\varphi_i(f)-\varphi_i(g)$ for all $i \in \N$ (\emph{including}
  $i=1$). Then we have two cases:
  \begin{enumerate}
    \item if $\varphi_1(f)-\varphi_1(g) \ge 0$, choose $j$ such that
      $\varphi_i(j) = \varphi_i(f)-\varphi_i(g)$ for all $i \in \N$, and choose
      $Z$ to be the empty set $\varnothing$. 
    \item if $\varphi_1(f)-\varphi_1(g) < 0$, choose $j$ such that $\varphi_1(j) =
      0$ and $\varphi_i(j) = \varphi_i(f)-\varphi_i(g)$ for all $i \in \N \setminus
      \{1\}$, and $Z$ to be a set of cardinality $\varphi_1(g) -\varphi_1(f)$. 
  \end{enumerate}
  We now have obtained $j \in \mor(\SetU)$ and $Z \in \lvert \SetU\rvert$ such
  that
  \[ 
    \varphi_i(f \sqcup 1_Z) = \varphi_i(g \sqcup j)
  \]
  for all $i \in \N$. Write $C$ and $D$ respectively for the domain and codomain
  of $j$. As $f \sqcup 1_Z$ and $g \sqcup j$ both have the same list of
  cardinalities of preimages of elements of their codomains, we may now choose
  bijections $\xi_1: A \sqcup C \to X \sqcup Z$ and $\xi_2: Y \sqcup Z \to B
  \sqcup D$ such that $ \xi_2 \circ (f \sqcup 1_Z) \circ \xi_1 = g \sqcup j$,
  as required.
\end{proof}

Write $\mathrm{FinSupp}_\N(X)$ for the set of finitely supported functions on
$X$; that is, the set of functions $f: X \to \N$ for which $f(x) \ne 0$ for at
most finitely many $x \in X$. Such functions are partially ordered by setting
$f \ge g$ if for all $x \in X$ we have $f(x) \ge g(x)$, and may be given a
commutative monoid structure by setting $(f+g)(x) = f(x)+g(x)$. 

The above discussion thus characterises the theory of resource convertibility
for $\pcd(\SetU,\BijU)$:

\begin{cor}\label{corollary1:PreorderOfBijUSetU}
We have an isomorphism of ordered monoids 
\begin{align*} 
  F: \big(\lvert\pcd(\SetU,\BijU)\rvert, \succeq, \sqcup\big) &\longrightarrow
  \big( \mathrm{FinSupp}_\N(\N\setminus\{1\}), \ge, + \big); \\
  f &\longmapsto \big(i \mapsto F_i(f)\big).
\end{align*}
\end{cor}

\subsection{The partitioned process theory of functions and injections}
Write $\InjU$ for the wide symmetric monoidal subcategory of $\SetU$ with
morphisms injective functions. We consider the partitioned process theory
$(\SetU,\InjU)$.

\begin{prop}\label{def:measureInjuSetU} 
  For $i \in \N$, define functions:
  \begin{align*}
    \gamma_{i}: \mor(\SetU) &\longrightarrow \N;  \\
    (f: X \to Y) &\longmapsto \#\big\{y \in Y \,\big\mid\, \# f^{-1}(y) \ge i\big\}.
  \end{align*}
  The family of monotones $\{G_i\}_{i \in \N\setminus\{0,1\}}$ induced by the
  family of functions $\{\gamma_i\}_{i \in \N\setminus\{0,1\}}$
  is a complete family of additive monotones for the resource theory
  $\pcd(\SetU,\InjU)$.
\end{prop}
The function $\gamma_i$ maps a function $f: X \to Y$ to the number of elements
of $Y$ that have at least $i$ elements of $X$ map to it; it is a sum of the
functions $\varphi_k$ for $k \ge i$.

\begin{proof}
  Theorem \ref{thm:main} again easily gives us that the $G_i$ are additive
  monotones for $i \in \N$, $i \ge 2$. In particular, note for condition
  \textit{(iii)} that pre-composing a function $f:X \to Y$ with an injection
  never increases the cardinality of the preimage of a point in $Y$, and that
  $f$ followed by an injection has preimages of points in the codomain that are
  either empty or equal to the preimage some point in $Y$, with no two points of
  the codomain sharing the same point in $Y$.

  It thus remains to prove the completeness of this family. Fix functions $f: X
  \to Y$ and $g: A \to B$, and suppose that for all $i \ge 2$ we have $G_i(f)
  \ge G_i(g)$. Again we wish to construct witnesses $Z \in \lvert \SetU\rvert$,
  $\xi_1,\xi_2 \in \mor(\InjU)$, and $j \in \mor(\SetU)$ such that $\xi_2 \circ
  (f \sqcup 1_Z) \circ \xi_1 = g \sqcup j$. There are many ways to construct
  such witnesses. We offer the following algorithm.

  Choose $Z$ to be a set of cardinality $\max\{0,\#B-\#Y\}$, $D$ to be a set of
  cardinality $\max\{0,\#Y-\#B\}$, and $j$ to be the unique function $j: \varnothing
  \to D$. This ensures that for all $i \in \N$, including $0$ and $1$,
  we have $\gamma_i(f \sqcup 1_Z) \ge \gamma_i(g \sqcup j)$. By definition, this
  means that for all $i \in \N$ we have
  \[
    \#\big\{y \in Y\sqcup Z\, \big\mid\, \# (f\sqcup 1_Z)^{-1}(y) \ge i\big\} \ge \#\big\{b \in B
      \sqcup D \,\big\mid\, \# (g \sqcup j)^{-1}(b) \ge i\big\}.
  \]
  This allows us to define an injection (in fact a bijection) $\xi_2: Y\sqcup Z
  \to B \sqcup D$ mapping each element $y \in Y \sqcup Z$ to an element $b \in
  B$ such that $\# (f \sqcup 1_Z)^{-1}(y) \ge \# (g\sqcup j)^{-1}(b)$. We then
  may choose an injection $\xi_1: A \to X \sqcup Z$ such that for all
  $a \in A$ we have $\xi_1(a) \in \big(\xi_2 \circ (f \sqcup 1_Z)\big)^{-1}(g(a))$.
  This proves the proposition.
\end{proof}

Analogous to the previous case, we reach the following characterisation of the
theory of resource convertibility for $\pcd(\SetU,\InjU)$.

\begin{cor}\label{corollary1:PreorderOfInjUSetU}
We have an isomorphism of ordered monoids 
\begin{align*} 
  G: \big(\lvert\pcd(\SetU,\InjU)\rvert, \succeq, \sqcup\big) &\longrightarrow
  \big( \mathrm{FinSupp}_\N(\N\setminus\{0,1\}), \ge, + \big); \\
  f &\longmapsto \big(i \mapsto G_i(f)\big).
\end{align*}
\end{cor}

\section{Some remarks on further directions}

While we have indicated how to construct additive monotones on resource theories
of parallel-combinable processes with discarding, there is work to be done to
understand their complete families better. In particular, an existence theorem or
otherwise for complete families of additive monotones would be of interest, as
well as a notion of minimally complete family of monotones. Some first steps
towards such results can be found in \cite[\textsection 6-7]{Fr}.

Observe also that the two examples of Section \ref{sec:ex} have an interesting
property: they in fact form a triple of inclusions of symmetric monoidal
categories 
\[
  \BijU \hooklongrightarrow \InjU \hooklongrightarrow \SetU.
\]
We might call this a doubly-partitioned process theory. This nested structure
seems to have been reflected in the construction of complete families of
additive monotones: we built one complete family from the other. We wonder
whether this could be done more generally.

A third salient question is the relationship between different methods of
constructing resource theories from partitioned process theories. A place to
start is perhaps to examine whether Theorem \ref{thm:main} has analogues for
related constructions.

\subsection*{Acknowledgements}
We thank Chris Heunen, Bob Coecke, Tobias Fritz, and Jamie Vicary for useful
conversations and encouragement, and an anonymous referee for detailed feedback
on a draft. BF thanks the Clarendon Fund and Hertford College, Oxford, for
their support. HNK thanks Wolfson College, Oxford and Mexico's National Council of
Science and Technology for their support.

\end{document}